\documentclass[%
 reprint,
 amsmath,amssymb,
 aps,
]{revtex4-2}
\usepackage{amsthm}
\usepackage[a4paper, total={7in,9in}]{geometry}

\newtheorem{lemma}{Lemma}
\newtheorem{corollary}{Corollary}
\newtheorem{theorem}{Theorem}
\newcommand{\St}{{\rm St}}
\newcommand{\Gr}{{\rm Gr}}

\newcommand{\X}{\mathcal{X}}

\newcommand{\Tr}{{\rm Tr}~}
\newcommand{\sk}{{\rm sk}}
\newcommand{\sy}{{\rm sy}}
\usepackage{graphicx}
\usepackage{dcolumn}
\usepackage{bm}

\begin{document}

\preprint{APS/123-QED}

\title{Many-Body Eigenstates from Quantum Manifold Optimization}

\author{Scott E. Smart}
\email{sesmart@ucla.edu}

\author{Prineha Narang}%
\email{prineha@ucla.edu}
\affiliation{%
College of Letters and Science, University of California, Los Angeles (UCLA), CA, USA.
}%

\date{\today}

\begin{abstract}
Quantum computing offers several new pathways toward finding many-body eigenstates, with variational approaches being some of the most flexible and near-term oriented. These require particular parameterizations of the state, and for solving multiple eigenstates must incorporate orthogonality. In this work, we use techniques from manifold optimization to arrive at solutions of the many-body eigenstate problem via direct minimization over the Stiefel and Grassmannian manifolds, avoiding parameterizations of the states and allowing for multiple eigenstates to be simultaneously calculated. These Riemannian manifolds naturally encode orthogonality constraints and have efficient quantum representations of the states and tangent vectors. We provide example calculations for quantum many-body molecular systems and discuss different pathways for solving the multiple eigenstate problem.
\end{abstract}


\maketitle

\section{Introduction}
The eigenstate problem is essential for simulating and understanding many-body quantum systems in chemistry and physics, with applications to other forms of state preparation. Over the past decade, several quantum computational approaches have been proposed for preparing the lowest-energy eigenstate\cite{Peruzzo2014,Head-Marsden2020,sajjanQuantumMachineLearning2022,linNearoptimalGroundState2020, Motta2019,mottaEmergingQuantumComputing2022a,  smartQuantumSolverContracted2021}, with predominant near-term approaches being related to the variational quantum eigensolver (VQE)\cite{tillyVariationalQuantumEigensolver2022,Peruzzo2014}. Approaches for finding \emph{many} eigenstates (which we refer to as the many-eigenstate problem) can require iterative constrained optimization, perturbative expansions, or solving a series of secular equations or generalized eigenvalue equation\citep{
higgottVariationalQuantumComputation2018,
smartManyBodyExcitedStates2023, ollitraultQuantumEquationMotion2020,
Huggins2020,
Colless2018}, and a recent line of variational approaches has emphasized that unitary transformations themselves preserve orthogonality\cite{nakanishiSubspacesearchVariationalQuantum2019a, xieOrthogonalStateReduction2022, yalouzStateaveragedOrbitaloptimizedHybrid2021, xuConcurrentQuantumEigensolver2023, hongQuantumParallelizedVariational2023}. A basic element of variational approaches is that the optimization is carried out with respect to a parameterization of the state, and can lead to several challenging issues, such as barren plateaus, high gradient and measurement costs, or spurious local minima due to the complex optimization surface\cite{McClean2018,anschuetzBarrenPlateausQuantum2022, tillyVariationalQuantumEigensolver2022,gonthierMeasurementsRoadblockNearterm2022}.

Vectors and collections of vectors with orthogonality or normalization constraints can be described using manifolds\cite{edelmanGeometryAlgorithmsOrthogonality1998a, Taylor1994}. Riemannian manifolds are differentiable sets that locally resemble Euclidean space and are equipped with a Riemannian metric\cite{boumalIntroductionOptimizationSmooth2023,IntroductionManifolds2008}. Given a problem on a Riemannian manifold, specific tools from manifold optimization
can be used to directly optimize elements of the manifold\cite{boumalIntroductionOptimizationSmooth2023,OptimizationAlgorithmsMatrix2007}. In particular, the Stiefel and Grassmannian manifolds describe sets of orthonormal vectors and have been used in machine learning, principal component analysis, mean-field quantum chemistry, and elsewhere\cite{shepardRepresentationParametrizationOrthogonal2015,
huckemannPrincipalComponentAnalysis2006,
turagaStatisticalAnalysisStiefel2008, edelmanGeometryAlgorithmsOrthogonality1998a,
chiumientoStiefelGrassmannManifolds2012}. Riemannian optimization has also found other recent applications in quantum sciences, including circuit optimization, state and channel tomography, and estimating quantum channel capacities\cite{luchnikovQGOptRiemannianOptimization2021,
wiersemaOptimizingQuantumCircuits2022,
luchnikovRiemannianGeometryAutomatic2021}.

In this work, we cast the subspace minimization and many-eigenstate problems as optimization problems on the Grassmannian and Stiefel manifolds respectively, which can be mapped to problems on quantum states. Using manifold optimization, we directly minimize states on the manifold, bypassing issues of parameterization present in most variational approaches. We provide a brief background on Riemannian optimization and then introduce representations of tangent spaces. We then show how the Stiefel manifold can be represented efficiently as a quantum state, and develop relevant tools for executing Riemannian optimization algorithms. We provide example calculations for many-body molecular systems and give an overview of different approaches for the many-eigenstate problem using classical and quantum resources. The current work establishes a framework for manifold optimization of quantum many-body eigenstate problems using quantum computation. 
\section{Theoretical Background}
We give a brief theoretical background for eigenvalue problems in quantum chemistry and physics, the Stiefel and Grassmannian manifolds, and then tools used in manifold optimization. Throughout, $\mathcal{H}^d$ refers to a finite-dimensional Hilbert space of dimension $d$, $\sk(A)$ and $\sy(A)$ respectively refer to the skew and symmetric portions of a square matrix $A$, $I_d$ refers the identity matrix of rank $d$, and $\otimes$ denotes the Kronecker product of two matrices. 
\subsection{Minimization Over $k-$Frames}
For a many-body quantum system, the Hamiltonian $\hat{H}$ describes all interactions within the system. For many problems, we can often express $\hat{H}$ in a finite-dimensional  Hilbert space, resulting in a Hermitian $n \times n$ matrix $H$. A common example is the many-body fermionic electronic Hamiltonian, which when expressed in the second quantization formulation consists of one- and two-body operators:
\begin{equation*}
\hat{H} = \sum_{pq} {}^1 H_{pq} a^\dagger_{p}a^{}_q + \frac{1}{2} \sum_{pqrs} {}^2 V^{pr}_{qs} a^\dagger_{p} a^\dagger_{r} a^{}_s a^{}_q
\end{equation*}
where $p,q,r,s$ denote sites (orbitals, modes, etc.), $a^\dagger$ and $a^{}$ denote the fermionic creation and annihilation operators, ${}^1 H$ contains the electron$-$nuclei and kinetic integrals, and ${}^2V$ is the repulsive electron$-$electron integrals. The basis consists of all exponential configurations of $N$ electrons in $m$ sites.

In this work, we consider real Hamiltonians (i.e. $H$ is symmetric), though this can be generalized to complex hermitian Hamiltonians that do not obey time-reversal symmetry. The many-body eigenproblem heuristically is to find a set of $k-$vectors which are the $k-$\emph{lowest} eigenvectors of $H$ (stationary solutions of the Schr{\"o}dinger equation). 

A collection of $k$ orthogonal vectors each of length $n$, where $1\leq k \leq n$, describes a $k-$frame (though for convention, we use $p$ instead of $k$ for the manifold dimension). We can write this collection of vectors as a $n \times p$ matrix $X \in \mathbb{R}^{n \times p}$. Then, the many-body eigenvalue problem can be written as:
\begin{equation}\label{eq:prob_st}
    f_{\St}(X) = \frac{1}{2} \Tr X^T H X K 
\end{equation}
where $K$ is a $p \times p$ diagonal matrix with distinct entries, determining their relative orderings. The $X$ that solves $f_S(X)$ has columns corresponding to the $p-$lowest eigenvectors of $H$. If we do not care about the ordering and simply want an invariant subspace of vectors, we can solve:
\begin{equation}\label{eq:prob_gr} 
 f_{\Gr}(X) =   \frac{1}{2} \Tr X^T H X .
 \end{equation} 
A solution $X$ of $f_{\Gr}(X)$ has the same span as a solution of Prob.~\eqref{eq:prob_st}, and is related by a $p\times p$ orthogonal transformation. The minimization problem also can be described as:
\begin{equation}\label{eq:minX}
X^* = \min_X  f(X)~~{\rm s.t.}~~ X^T X  = I_p
\end{equation}
where $I_p$ represents the identity on $\mathcal{H}_p$. Note, if the elements of $K$ are \emph{not} distinct, then we can form separate subspaces based on the degeneracies of the diagonals (e.g., $K$ with diagonal entries $\pm1$ will result in a maximal and minimal subspace).

These two problems represent well-known optimization problems over the Stiefel and Grassmannian manifolds\cite{edelmanGeometryAlgorithmsOrthogonality1998a}, which we discuss in the next section. Conceptually, the Grassmannian manifold represents a subset of the Stiefel manifold, and problems on Grassmannian are typically made as simplifications of the Stiefel manifold. However, computationally the Grassmannian problem, Pr.~\eqref{eq:prob_gr}, often yields simpler tools, and we show later, can be used to solve the Stiefel problem, Pr.~\eqref{eq:prob_st}.

\subsection{Stiefel and Grassmannian Manifolds}
Manifolds are objects which locally resemble Euclidean space. For a $n-$sphere in $\mathbb{R}^n$ for instance, we have that $x^T x-1 =0$, and at each point $x$ on the sphere we can define a linear subspace $V = \{ v : v^Tx = 0\} $, which is the tangent space. The Stiefel and Grassmannian, are embedded submanifolds of $\mathbb{R}^{np}$ and exhibit several special features. Several excellent references on manifolds, the Stiefel and Grassmannian manifolds, and manifold optimization exist\cite{IntroductionManifolds2008,boumalIntroductionOptimizationSmooth2023,
edelmanGeometryAlgorithmsOrthogonality1998a,
OptimizationAlgorithmsMatrix2007}.

The real Stiefel manifold is defined as:
\begin{equation}
    \St(n,p) = \{X \in \mathbb{R}^{n\times p} : X^T X = I_p \}
\end{equation}
which otherwise is the set of all possible $k-$frames. The Grassmannian manifold is a quotient manifold of the Stiefel manifold with respect to orthogonal transformations of the $p-$elements. That is, given an equivalence relation $\sim$ between elements in $\St(n,p)$:
\begin{equation*}
    X \sim Y \iff {\rm span}(X) = {\rm span}(Y),
\end{equation*}
where the span is the column span of the matrix, 
we can define equivalence classes
\begin{equation*}
    [X] = \{ Y \in \St(n,p) : Y \sim X \}.
\end{equation*}
The Grassmannian manifold then is the quotient set of the $\St(n,p)$ w.r.t. $\sim$, that is:
\begin{equation}
    \Gr (n,p) = \St(n,p) / \sim  ~ = \{ [X] : X \in \St(n,p) \}.
\end{equation}
which is the set of all equivalence classes of $k-$frames. Complex variants of these manifolds exist, though in this work, we focus on the real case, which is sufficient for a large class of many-body Hamiltonians.  
For $p=1$, the Grassmannian and Stiefel problems are equivalent, and for $p=n$, the Grassmannian case is trivial and the Stiefel manifold is akin to optimization over the orthogonal group.

Having defined the manifolds of interest, the optimization problem in Eq.~\eqref{eq:minX} can be rewritten as:
\begin{equation}
X^* = \min_{X \in \mathcal{M}} f_\mathcal{M}(X)
\end{equation}
where the manifold $\mathcal{M}$ is either $\St(n,p)$ or $\Gr(n,p)$. This problem can be solved directly using manifold optimization.  
\subsection{Manipulating Riemannian Manifolds}

\begin{figure}
\includegraphics[scale=0.19]{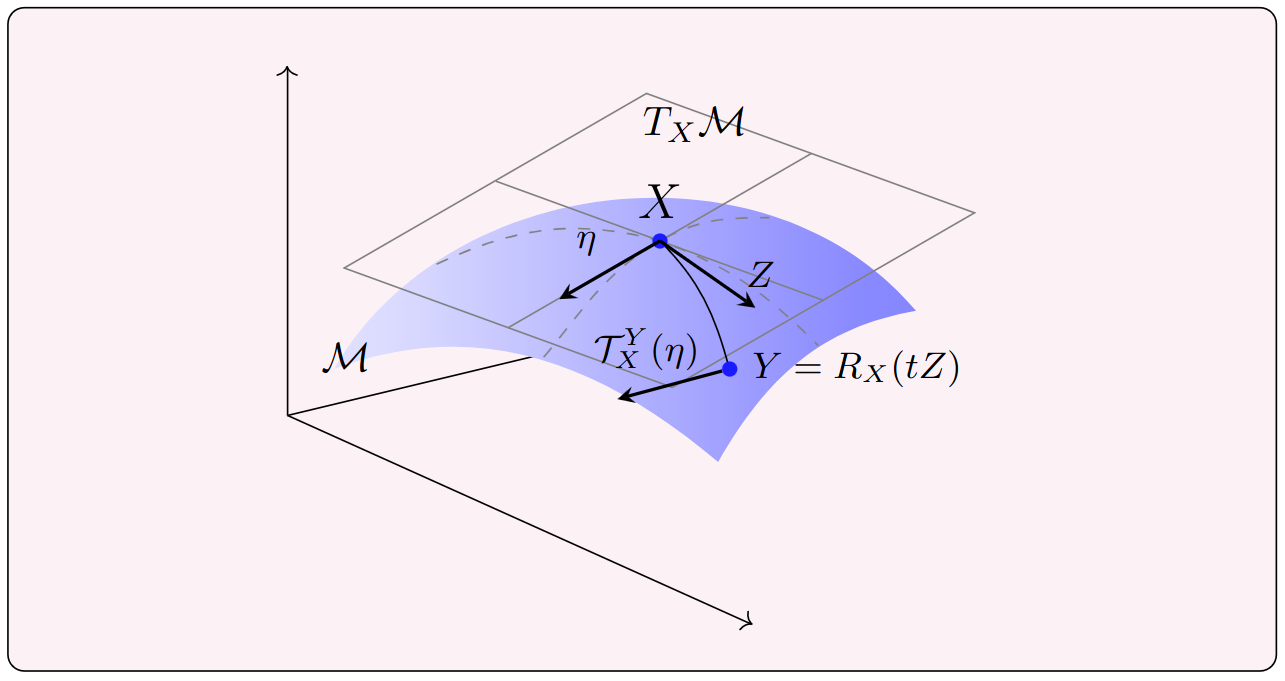}
\caption{Selected elements of Riemannian manifolds. Given an element $X$ of the manifold, the tangent space $T_X M$ contains elements perpendicular to the manifold. We can move along the manifold to a point $Y$ using a retraction $R_X(tZ)$ in the directions of a tangent vector, and we can transport tangent vectors using transporters $\mathcal{T}_{R_X}$.}
\label{fig:manifolds}
\end{figure}

Movement on the manifold is often described by parameterized curves. Let $c(t)$ be a curve such that $c(0)=X$ and $\dot{c}$ describes the velocity of the trajectory in a certain direction. The tangent space at a point $X$ is given by the collection of all possible $\dot{c}$ elements. From the Stiefel manifold, where  $X^T(t)X(t)=I_p$, the tangent space for the Stiefel manifold at $X \in \St(n,p)$ is the set:
\begin{equation}
    T_X \St(n,p) = \{Z : Z^T X + XZ^T= 0  \}  ,
\end{equation}
which is the collection of $n \times p$ matrices $Z$ such that $Z^T X $ is skew-symmetric. The tangent space of the quotient manifold has a more complicated structure and consists of a vertical portion containing tangent elements within the equivalence class, and then a horizontal tangent space, which generally is the tangent space of interest and has a unique representative in the Stiefel tangent space \cite{boumalIntroductionOptimizationSmooth2023,edelmanGeometryAlgorithmsOrthogonality1998a}:
\begin{equation}
    T^H_{X} \Gr(n,p) = \{ Z : Z^T X = 0  \}. 
\end{equation}
Projections onto the tangent space are useful within a variety of contexts, and both manifolds have common projections:
\begin{align}
    P^{\St}_X(Z) = Z - X ~\sy X^T Z ,  \\
    P^{\Gr}_X(Z) = Z - XX^T Z . 
\end{align}
If at every point $X$ we endow the tangent space with a Riemannian metric, then we have a Riemannian manifold. A natural metric for embedded submanifolds is the Euclidean inner product for matrices:
\begin{equation}
    \langle Z, Y \rangle_X = \Tr Z^T Y
\end{equation} 
where $Z,Y \in T_X M$. Because it is generally assumed that the inner product is at $X$, we often omit the subscript. 

To move around the manifold, we use retractions, which are mappings from the tangent space to the manifold. A retraction at a point $X \in M$ in a direction $V \in T_X M$ is a function $f_X(V)$ that satisfies:
\begin{align*}
f_X(0) &= X , ~~\frac{d}{dt} f_X(tV)|_{t=0}= V .
\end{align*}
This represents a first-order retraction, and higher-order retractions (such as the second-order retraction with constant acceleration) can be derived. 

\subsection{Riemannian Optimization}
To consider how solutions to the problems on the manifold, we first need to establish the geometry of the problem through the Riemannian gradient and Hessian.  

For a function $f$ on a Riemannian manifold let $\bar{f}$ denote its smooth extension on $\mathbb{R}^n$. The Riemannian gradient ${\rm grad} f$ is given as the projection of the gradient of $\bar{f}$ onto the tangent space:
\begin{align}
{\rm grad} f(X) = P_X ({\rm grad} \bar{f}(X))
\end{align}
and satisfies (given some retraction $R_X$:
\begin{equation*}
f(R_X(t V )) = f(X) + t \langle {\rm grad} f(X), Z \rangle +  O(t^2). 
\end{equation*}
For a tangent vector on the manifold, we can calculate its vector field derivative or affine connection $\nabla_U$. The standard connection for embedded submanifolds is the tangent projection of the smooth vector field derivative:
\begin{align}\label{eq:rhess}
{\rm Hess} f(X)[V] = \nabla_V G(X) = P_X (D\bar{G}(X)[V]).
\end{align}
Here $G$ is an extension of ${\rm grad} f(X)$ in real space, and $D$ is the differential. We show explicit forms for Problems \eqref{eq:prob_st} and \eqref{eq:prob_gr} in the next section. Regardless, these allow us to describe critical points and second-order geometry of the manifold and have important relations to the Euclidean gradient and Hessian. 

The last component involves vector transport, which allows us to move vectors along the manifold. Due to the structure of the manifold (described by the connection), a function transporting a vector may not necessarily be unique (unlike parallel transport). Given points $X$ and $Y \in \mathcal{M}$, where $Y = R_X(tZ)$, a transporter $\mathcal{T}^{Y}_X(\eta)$ is a linear map from $T_X$ to $T_Y$ where $\mathcal{T}_X^X(\eta)=\eta$. 

The projection operator is an example of a transporter, and one can sometimes use retractions to build transporters. In an optimization procedure, vector transport allows us to update information in the search direction related to previous steps. The above points are shown pictorially in Fig.~\eqref{fig:manifolds}.

\section{Representing Tangent Vectors of the Stiefel Manifold}\label{sec:tangents}

Before introducing the quantum scheme, we take a closer look at tangent vectors and our choice of representation. While we will have a relatively straightforward way to represent the Stiefel manifold on a quantum device, we need a suitable means to find and manipulate tangent vectors on a quantum computer. The key idea is that we use operators associated with tangent vectors instead of tangent vectors themselves. 

Specifically, we find left and right Lie algebra actions on the state (which we will commonly refer to as left and right actions)\cite{hallLieGroupsLie2015a}. These Lie algebra actions are elements of the skew-symmetric matrix Lie algebra $\sk(n)$ and $\sk(p)$, which act on the manifold to yield elements of the corresponding tangent space. By defining an inner product on these actions, we can substitute the tangent vectors for actions, and then define suitable retractions. 

\subsection{Tangent Vector Decompositions}
A common tangent vector decomposition for the Stiefel manifold involves the projection on $X$ and its orthogonal complement. Specifically, let $X_\perp$ be a set of $n \times (n-p)$ matrices such that $I_n = XX^T + X_\perp X_\perp^T$. Then, we can decompose $Z \in T_X \St(n,p)$ as:
\begin{align}
Z =  X A  + X_\perp C
\end{align} 
Here, $A$ has $n(n-1)/2$ degrees of freedom and  $X_\perp^T Z$ is $(n-p) \times p $ dimensional, representing the degrees of freedom in the tangent space. 

The tangent vector also yields the following decomposition:
\begin{equation}
Z = L_X(Z) X - X A_X(Z)
\end{equation}
where $L_X$ and $A_X$ are mappings from the tangent space to skew-symmetric matrices of dimension $n$ and $p$ respectively, and are given as:
\begin{align}\label{eq:leftact}
L_X(Z) &= Z X^T - XZ^T \\ 
A_X(Z) &= X^T Z . \label{eq:rightact} 
\end{align}
For the Stiefel manifold these represent left and right Lie algebra actions which when applied to the manifold yield an element of the tangent space. For the Grassmannian manifold, it holds that $A_X(Z)=0_p$. 

Left multiplication of $Z$ by $X^T$ yields the relation:
\begin{equation}\label{eq:left_right}
A_X(Z) = \frac{1}{2} X^T L_X(Z) X
\end{equation}
indicating that there exists a further parametrization of the skew-symmetric degrees of freedom:
\begin{equation}\label{eq:left_general}
L_X^\alpha(Z) = (I - \alpha XX^T)ZX^T - XZ^T (I-\alpha XX^T)
\end{equation}
where $\alpha \in [0,1]$.  This decomposition has been noted in different contexts, namely through variations of the Stiefel inner product \cite{huperLagrangianApproachExtremal2021a, zimmermannComputingRiemannianLogarithm2022b}. In this work, we focus on the $\alpha=0$ case, which practically can lead to simpler expressions for the problems of interest, though in the Appendix we explore the case where $\alpha=1$.

A natural retraction of $L_X(V)$ and $A_X(V)$ is given by: 
\begin{equation}\label{eq:retract}
R_X(V) = e^{L_X(V)} X e^{-A_X(V)}.
\end{equation}
For the more general form in Eq.~\eqref{eq:left_general}, the corresponding retraction is given by:
\begin{equation}
R^{\alpha}_X(V) = e^{L_X^\alpha(V)} X e^{(2\alpha - 1)A_X(V)}.
\end{equation}
These are both in principle second-order retractions, though practically they are limited by the accuracy to which we can implement the matrix exponential. 

\subsection{Inner Products on Tangent Vector Operators}

The actions given in Eqs.~\eqref{eq:leftact} and \eqref{eq:rightact} have inner products which we can define as well. 

\begin{lemma}\label{lem:metric}
 Given two tangent vectors $Z$ and $W$ in $T_X \St(n,p)$, the inner product $\langle Z, W \rangle$ can be defined in terms of inner products on their respective left and right transformations:
\begin{align}
\langle Z, W\rangle =  \frac{1}{2}\langle L_X(Z), L_X(W) \rangle - \langle A_X(Z), A_X(W) \rangle,
\end{align}
where $\langle \cdot , \cdot \rangle$ is the Euclidean trace product. 
\end{lemma}
\begin{proof}
The above relation can be derived using $L_X(Z) = X_\perp C X^T - X C^T X_\perp + 2 XA_X(Z)X^T $ on the right-hand side, shown in the Appendix. \end{proof}
Again, the Grassmannian result holds by setting $A_Z$ to $0_p$.

\subsection{Generating Actions from Tangent Vectors}
While for a tangent vector $Z$, the mappings $L_X$ and $A_X$ can be found directly, we would like some insights on how these transformations behave with generic elements of $\mathbb{R}^{n \times p}$ and tangent space projections. 
\begin{lemma}\label{lem:st}
For a $n \times p$ matrix $V$, $L_X(V)$ is equivalent to $L_X(P^{\St}_X(V))$.
\end{lemma}

\begin{lemma}\label{lem:gr}
For a $n\times n$ matrix $B$ acting on $X$, the left action of the Grassmannian tangent projection of $BX$ is given by the skew symmetric matrix:
\begin{align}\label{eq:thm1}
\begin{split} L_X (P_X^\Gr(BX)) &= 
[\sy(B), XX^T ]  +\{\sk(B), XX^T  \} \\ &~~ - 2 XX^T  \sk(B) XX^T .
\end{split}\end{align}
\end{lemma}
\begin{proof}
The left action of $BX$ itself yields 
\begin{align}\label{eq:leftdecomp}
L_X(BX) =  [\sy (B), XX^T ] + \{ \sk B, XX^T \}.
\end{align}
Inserting the projection $P_X^\Gr(BX)=[B,XX^T]X$ yields Eq.~\eqref{eq:thm1}. 
\end{proof}
This has an immediate corollary, relating symmetric operators acting on $X$ to their projection onto the tangent space. 
\begin{corollary}\label{cor:symm} The commutator of a symmetric matrix $A$ with $XX^T$, $[A,XX^T]$, is a left action on $X$ producing a tangent vector in $T_X \Gr(n,p)$.
\end{corollary}
Looking at the right action, we have a straightforward result. 
\begin{lemma}\label{lem:right}
The transformation $A_X$ for a general $n \times p$ matrix $W$ projected onto the Stiefel manifold is:
\begin{equation}
A_X(P^\St_X(W)) = \sk(X^T W).
\end{equation}
\end{lemma}

Let $W= BXK$ with $B$ a $n \times n$ matrix and $K$ $p \times p$ symmetric. Then, the left and right actions are given as
\begin{align}
L_{X}(BXK) &=   [\sy (B), XKX^T ] + \{ \sk(A), XKX^T \},
\end{align}
\begin{equation}
\begin{split} A_X(P_X^\St(W)) =  & \frac{1}{2} \{ X^T \sk (B) X, K\}  +\frac{1}{2}[X^T \sy(B) X, K] . 
\end{split}
\end{equation}

\subsection{Riemannian Gradient and Hessian Forms} 
Using the above results, we can develop gradient and Hessian forms for our problems of interest.
\subsubsection{Subspace Invariant Minimization}
For the subspace minimization problem on the Grassmannian manifold, tangent vectors can be expressed using only the left action. Thus, the left action of the Riemannian gradient of Problem~\eqref{eq:prob_gr} at $X$ and the Hessian acting in a direction $V \in T_X \Gr(n,p)$ at $X$ are:
\begin{align}
L_X({\rm grad} f_\Gr(X)) &= [H, XX^T], \\
L_X({\rm Hess}f_\Gr(X)[V]) &= [[H, L_X(Z)], XX^T].
\end{align}
The gradient follows directly from the left action acting on the Grassmannian projection, as in Lemma \eqref{lem:gr} and Corr.~\eqref{cor:symm}~above, and the Hessian form is derived in the Appendix. 
\subsubsection{Generic Eigenvector Minimization}
For the eigenvector problem on the Stiefel manifold, we can define a similar Riemannian gradient:
\begin{align}
L_X({\rm grad~}f_\St(X)) &= [H, X K X^T], \\
  A_X({\rm grad~}f_\St(X)) &= \frac{1}{2} [X^T H X, K].
\end{align} 
The Hessian of $f_\St(X)$ does not appear to have a simple form in terms of the left and right action of a tangent vector, though we derive one possible form in Appendix \ref{sec:stiefhess}.

\section{Representing Stiefel Manifolds as Quantum States}

In this section, we show how we can use quantum computers to store and manipulate Stiefel and Grassmannian elements, as well as the left and right actions of tangent vectors. We highlight the general structure of our approach in Fig.~\eqref{fig:stiefcirc}, showing state preparation, retractions, and operator measurement.

\begin{figure}
\includegraphics[scale=0.20]{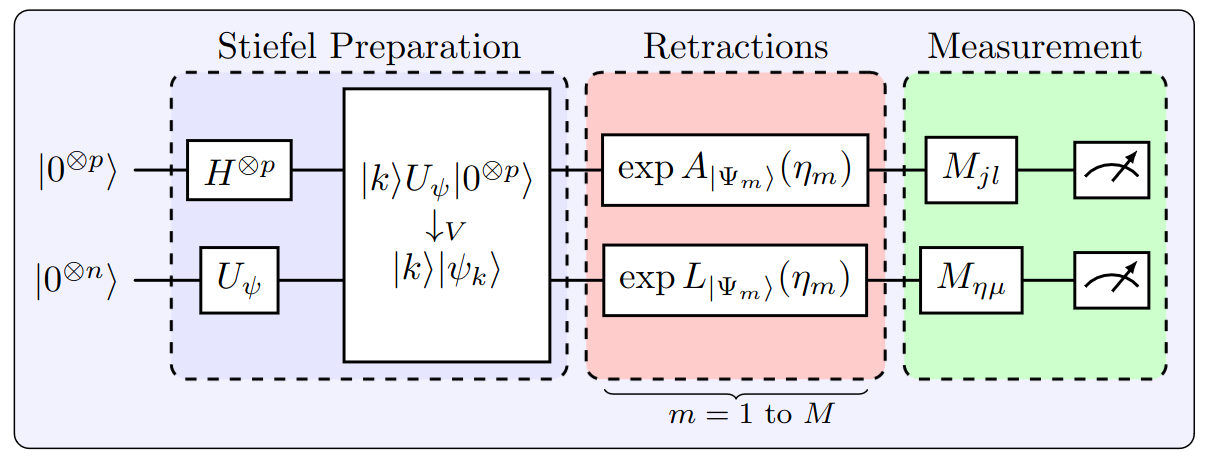}
\caption{Quantum circuit structure for preparing elements of the Stiefel and Grassmannian manifold on a quantum computer. From the computational basis, we create a maximally entangled state with a set of ancilla states and a set of reference states. Retractions are carried out iteratively (here, $M$ times) on the two subsystems, and measurement (relative to some basis $M_{ij}$) allows us to capture observables of interest.} 
\label{fig:stiefcirc}
\end{figure}

\subsection{States and $k-$Frames}
Vectorization of a $n\times p$ matrix $X $ arranges the columns of $X$ into a single vector. In the case of the $k-$frames represented in the Stiefel and Grassmannian optimization, columns represent vectors $|\psi_j\rangle \in \mathcal{H}^n$, with a corresponding computational index, with vectorization yielding: 
\begin{equation}
{\rm vec}(X) = \sum_{k=0}^{p-1} |k \rangle \otimes |\psi_k \rangle 
\end{equation}
where we can take $|k\rangle$ simply to be the binary representation of $k$. The vectorization operator has the useful properties:
\begin{equation}
{\rm vec}(ABC) = C^T \otimes A ~{\rm vec } {(B)} 
\end{equation} as well as the Euclidean trace product of two matrices being equivalent to their vectorized inner product. We can define a state $|\Psi\rangle$ corresponding to ${\rm vec}(X)$:
\begin{equation}
|\Psi \rangle = 2^{-\frac{p}{2}}\sum_{k=0}^{p-1} |k\rangle |\psi_k \rangle = 2^{-\frac{p}{2}}{\rm vec}(X)
\end{equation}
which is a maximally entangled state. The key advantage is that $X$ is classically of dimension $np$, whereas $\Psi\rangle$ requires only $\log_2 n + \log_2 p $ qubits, an exponential reduction in storage cost. This ansatz form has been recognized elsewhere \cite{xieOrthogonalStateReduction2022,xuConcurrentQuantumEigensolver2023, hongQuantumParallelizedVariational2023} (and is in part inspired by observations with the subspace search variational quantum eigensolver \cite{nakanishiSubspacesearchVariationalQuantum2019a}), though not within the context of manifold optimization. 

\begin{table*}
\caption{Comparison of proposed quantum and established classical manifold optimization tools for optimizing over the Stiefel and Grassmannian manifolds. Here, $Z$ and $W$ are tangent vectors, with $A_Z = A_X(Z)$ and $L_Z = L_X(Z)$. Classically, numerous retractions exist, such as the Cayley transform or QR decomposition\cite{boumalIntroductionOptimizationSmooth2023}.} 
\label{tab:components}
\begin{tabular}{c|cc|cc}
Components & Classical $\Gr(n,p)$ & Quantum $\Gr(n,p)$ & Classical $\St(n,p)$& Quantum $\St(n,p)$ \\ \hline
Storage  & $np$ & $\log_2 np$ & $np$ & $\log_2 np$ \\
Tangent Representation & $\mathbb{R}^{(n-p)p}$ & $\sk(n)$ & $\mathbb{R}^{p(n-p)} $ & $\sk(n) +\sk(p)$ \\ 
Retraction & Many${}^*$ &  Exp.  & Many${}^*$ & Exp. \\
Vector Transport & $R_X[V]$ & BCH & $R_X[V]$ & BCH \\  
Metric & $\langle Z,W \rangle$ & $\langle L_Z, L_W\rangle$ & $\langle Z,W \rangle$ & $\langle L_Z, L_W\rangle$, $\langle A_Z, A_W\rangle$ \\ 
\end{tabular}
\end{table*}

In principle, the operations required for traditional manipulation of the manifolds (primarily matrix-vector and rank-$p$ multiplication) can be implemented on quantum computers, and thus common tools for the Stiefel and Grassmannian manifolds could be carried out. The primary difficulty here is that generic linear algebraic operations require substantial ancillary costs, and so are largely unfeasible in non-fault-tolerant quantum computation. Additionally, in contrast to elements of the manifold, the varying normalization of tangent vectors represents a challenge. Compounding this is that optimization requires iterative manipulation, implying qubit costs scaling with the number of iterations and extensive amplitude amplification to counteract necessary dilations. 

We show however that the tools required for manifold optimization with the Stiefel and Grassmannian manifolds can also be realized using more appropriate near-term tools, using classical and quantum resources and the left and right actions developed previously. Table~\ref{tab:components} gives a summary of the different components.

\subsection{Measurement and Tomography}
Measurement on the Stiefel manifold corresponds to observable estimation and state tomography of operators on the $n$, $p$, and $n \times p$ dimensional Hilbert spaces, which we describe below. 

The most basic problem relates to function evaluation. For Problem~\eqref{eq:prob_gr}, this is:
\begin{align}
\frac{1}{2} \Tr X^T H X &= 2^{p-1} \langle \Psi | I_p \otimes H | \Psi \rangle ,
\end{align}
and for Problem~\eqref{eq:prob_st}, we have:
\begin{align}
\frac{1}{2} \Tr X^T H X K &= 2^{p-1} \langle \Psi | K \otimes H | \Psi \rangle 
\end{align}
which both can be measured using standard estimation techniques using local operator decompositions of $H$ and $K$, i.e. with Pauli matrices. 

The matrix product $XX^T$ has the form $XX^T = \sum_k |\psi_k \rangle \langle \psi_k |$, which is the reduced density matrix obtained by tracing over the index degrees of freedom. Letting $\eta, \mu$ be elements of a complete basis of $\mathcal{H}^n$, $BXX^T$ has the form:
\begin{align}
BXX^T =  2^{p}\sum_{\eta, \mu} \langle \Psi | I_p \otimes  (M_{\mu \eta} B) | \Psi \rangle M_{\eta\mu} 
\end{align}
where $M_{\mu \eta } = | \mu \rangle \langle \eta|$. If we replace $ XX^T$ with $X K X^T$, where $K$ is not necessarily diagonal, then we have:
\begin{align}
B  X K X^T &=  2^{p}\sum_{\eta, \mu}  \langle \Psi | K \otimes (M_{\mu \eta} B) | \Psi \rangle M_{\eta\mu} .
\end{align}

Finally, we can measure $p \times p$ operators themselves as follows. Letting $|k\rangle$ denote a basis over $\mathcal{H}^p$, and a $p\times p$ matrix $X^T B X$, we have:
\begin{align}\label{eq:anc_meas}
X^T B X &= 2^{p} \sum_{kj}  \langle \Psi | M_{kj} \otimes B | \Psi  \rangle  M_{kj}
\end{align}
where $M_{kj} = | k \rangle \langle j |$. When $B=H$, this gives the subspace energy matrix. Collectively, these give us a complete set of tools to measure and construct tangent operators, including all of the quantities required in the derivatives of Problems \eqref{eq:prob_st} and \eqref{eq:prob_gr}. 

\subsection{Retractions on the Stiefel and Grassmannian Manifold}
For a tangent vector $Z \in \St(n,p)$ the vectorization of the tangent vector decomposed into the left and right actions is:
\begin{equation}
{\rm vec}(Z) = (I_p \otimes L_X(Z) + A_X(Z) \otimes I_n )~{\rm vec}(X).
\end{equation}
which implies a simple retraction on $|\Psi\rangle$ using the linearity of the exponential transformation:
\begin{equation}\label{eq:qretract}
R_X(tZ) = \frac{1}{2^{p/2}}e^{t A_X(V)} \otimes e^{ t L_X(Z)} |\Psi \rangle.
\end{equation}
The given retraction is a second-order retraction but classically is challenging to implement for large dimensions. Implementing matrix exponentials on quantum devices is an active area of research and in general, can be most readily realized in an ancilla-free manner via trotterization techniques\cite{childsTheoryTrotterError2021b}. Practically, we find that a second-order approximation is sufficient for an accurate Hessian, although a first-order trotterization is sufficient for convergence guarantees\cite{boumalIntroductionOptimizationSmooth2023}.

Given some operator $O \in \sk(n)$, $O$ acting on $X$ will always be associated with some tangent vector on the Stiefel manifold, though possibly not the Grassmannian. To remove the orthogonal degrees of freedom, we can measure the operator $O$ coupling to the ancilla, via Eq.~\eqref{eq:anc_meas}, and apply the degrees of freedom in the retraction: 
\begin{equation}
R_X(P^\Gr_X(OX)) = e^{O} \otimes e^{X^T O X}{\rm vec}(X).
\end{equation}
We can also use such a technique to check or confirm that a vector $O X$ belongs to $T_X\Gr(n,p)$.

\subsection{Vector Transport}
As previously mentioned, transporters allows us to move vectors between tangent spaces, though is not uniquely defined. Here, we can consider a transporter that specifically transports the actions of a tangent vector. For instance, a tangent vector $\eta$ is transported to $Y=R_X(Z)$ with the exponential retraction as follows:
\begin{align}
\begin{split}
\mathcal{T}_X^{Y}(L_\eta X - X A_\eta) = e^{L_Z} L_\eta e^{-L_Z} Y - Y e^{A_Z}
A_\eta e^{-A_Z} \end{split}.
\end{align}
One can easily verify that the resultant vector is in the tangent space at $Y$. We can then approximate the transported left and right actions (at $Y$) through the BCH expansion:
\begin{equation}
L_{T_X^Y(\eta)}(Y) = L_\eta + [L_Z, L_\eta] + \frac{1}{2}[L_Z,[L_Z, L_\eta]] + O(L_Z^3),
\end{equation}
and similarly for the right action. For small $p$, we also can apply retractions exactly if necessary.

\section{Numerical Calculations}
We highlight the above ideas with noiseless quantum simulations utilizing quantum computational tools. Because the ideas here are quite general, different optimization strategies, ansatz implementations, and retractions can be used. For instance, the commutator expressions present in the gradient and Hessians of Problems \eqref{eq:prob_st}-\eqref{eq:prob_gr} can be obtained from methods based on finite differences using exponential transformations. In other contexts, such as with the contracted quantum eigensolver\cite{smartManyFermionSimulationContracted2022}, this can reduce the number of circuits at a potentially increased measurement cost. Another route for measuring tangent vectors involves estimation via shadow tomography\cite{huangPredictingManyProperties2020d}, allowing us to reconstruct these with a cost dependent on the corresponding shadow norm and desired accuracy. 

The current work utilizes custom packages interfacing with Qiskit (v0.45.1) designed for state preparation and operator manipulation. The essential tools are mappings from operators to exponential unitary transformations and measurement of the state. The Hamiltonian was generated from electron integrals generated with PySCF (v2.2.1, \cite{Sun2017}). The Riemannian gradients and Hessian expressions with the retractions were numerically verified using the recommendations in Ref. \cite{boumalIntroductionOptimizationSmooth2023}.

\subsection{Grassmannian Example with the Riemannian Newton's Trust Region Approach}

For the Grassmannian case, we take the square geometry of molecular $H_4$ at $1.0$~\AA~in a minimal (STO-3G) basis set.  Using the Jordan-Wigner transformation with parity trimming, we obtain a 6-qubit Hamiltonian. We prepare manifold states from $0$ to $3$ ancilla qubits, corresponding to $\Gr(256,1)$ to $\Gr(256,8)$ (though in our symmetry subspace, there are only 36 states of interest). We screen the mean-field state of the Hamiltonian, choosing the $p$-lowest entries, which can be prepared using controlled swap operations. The gradient and Hessian forms are calculated directly from the set of measured observables, and we use the set of 4-RDM operators (which for a four-particle system are equivalent to the computational elements). Using a second-order trotterization, we can implement the left actions of a tangent vector $\eta$ on a quantum computer as an approximation of the exponential function.

For the optimizer,  we use the Riemannian Newton's trust region method\cite{boumalIntroductionOptimizationSmooth2023}, which does not require vector transport or Hessian storage. Instead, we use the Riemannian conjugate gradient (RCG) approach to solve Newton's equation in a trust region. Here, we limit the number of RCG steps to 3, and our initial trust radius is $0.25$. At each step, we attempt to find a search direction $\eta$ such that \begin{equation}
{\rm Hess} f_G(X)[\eta] = - {\rm grad}f_G[X] 
\end{equation} 
subject to a trust region constraint ($\langle \eta, \eta \rangle \leq \kappa $).  This does not involve storage of a Hessian-like quantity but does require calling the Riemannian Hessian at each RCG step. Figure~\eqref{fig:grass} depicts our results, notably the total overlap of each ideal eigenstate with the subspace at each step of the optimization. In each case, we obtain rapid convergence to the appropriate solution in only a handful of steps, following the expected second-order convergence of Newton's method. 

\begin{figure}
\includegraphics[scale=0.51]{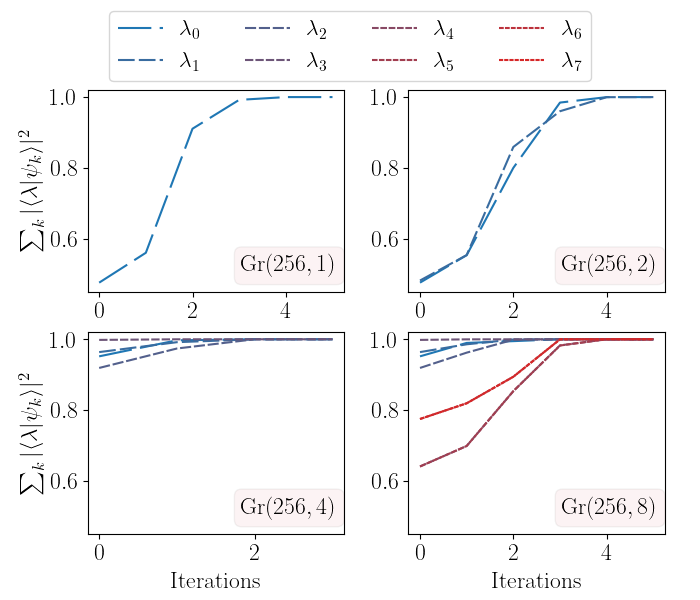}
\caption{Sum of eigenstates populations represented in the subspace, with an increasing number of states included in the Grassmannian optimization. Higher eigenvalues are represented by dashes and red coloring. We used a threshold of 0.001 in the gradient norm. For each state, we can recover the proper energies in no more than 6 iterations. }
\label{fig:grass}
\end{figure}

\subsection{Stiefel Optimization with the Riemannian Conjugate Gradient}

For the second calculation, we look at linear $H_3$, which can also be mapped using the Jordan-Wigner transformation. This can be considered as an optimization over $\St(64,4)$, although there are only 9 states in the $N=3$ and $S_z=\frac{1}{2}$ symmetry subspace of interest. By selecting $K$ to be diagonal with $p-$non degenerate eigenvalues, the optimizer should correctly identify individual states. 

We used the Riemannian non-linear conjugate gradient approach\cite{satoRiemannianConjugateGradient2022}, which at each step obtains a search direction conjugate to previous steps. Unlike Newton's method, we do require transport of tangent vectors $\eta$ and ${\rm grad} f_S(X)$ for the calculation of $\beta$ and the search direction, although for our retraction, $L_X(\eta)$ and $A_X(\eta)$ are invariant under the respective actions. We use the Riemannian form of the Hestenes-Stiefel update to calculate $\beta$.

We consider two forms of $K$. In the first, $K$ is a diagonal positive matrix with increasing diagonals ($K_{ii} > K_{i+1,i+1} > 0$). In the second, we simply invert the orderings, mapping the $i$-th entry of $K$ to the $(p-i+1)$-th entry. The results can be seen in Fig.~\ref{fig:stiefel}.

In the first case, all of our states have some overlap with their target state, and within 5 iterations each state has nearly converged. In contrast, by swapping our orderings, several iterations are required to reorder the states. The ordering of these states is related to the relative eigenvalues of the particular states in question and can bias or influence the practical speed of the calculation\cite{edelmanGeometryAlgorithmsOrthogonality1998a,hongQuantumParallelizedVariational2023}. 

\begin{figure}
\includegraphics[scale=0.51]{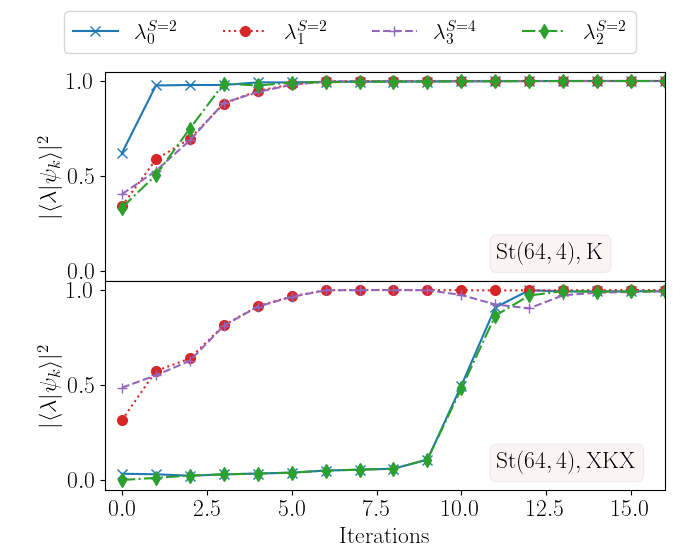}
\caption{State overlaps for the first four eigenstates of linear $H_3$ obtained in an optimization over the Stiefel manifold using the Riemannian conjugate gradient method, with the diagonals of $K$ (top) equal to $(4,3,2,1)$, and $XKX$ (bottom) implying the flipped diagonal. The same initial states are used, showing the algorithm can find and reorder solutions by manipulating orthogonal degrees of freedom.}
\label{fig:stiefel}
\end{figure}

\section{Solutions to the Stiefel Minimization Problem}
Noting that a solution to the Grassmannian problem is a necessary but not sufficient criteria for a solution to the Stiefel problem, we discuss four potential solutions to the Stiefel problem using classical and quantum techniques.
 
{\bf 1. Minimization over the Stiefel Manifold:} The first approach involves directly solving the Stiefel optimization problem, as in the second example. This does not present a theoretical challenge, as generic operators $K$ can be implemented in principle on a quantum device up to a scaling factor. However, practically there are limitations and questions as to how the different states should be weighted. Furthermore, these additional terms involve more complex measurements which may or may not be friendly for near-term devices. 

{\bf 2. Classical Diagonalization of the Energy Matrix: }The second solution involves the minimization of the Grassmannian problem, followed by a classical diagonalization of the subspace energy matrix $X^T H X$. This is similar to the generalized eigenvalue problem, although the eigenvectors in principle will be exact, and there is no potential indeterminacy in $X$. The main challenge here is for large $p$. A demonstration of this approach is shown in Fig.~\eqref{fig:grass_flow}.

\begin{figure}
\includegraphics[scale=0.32]{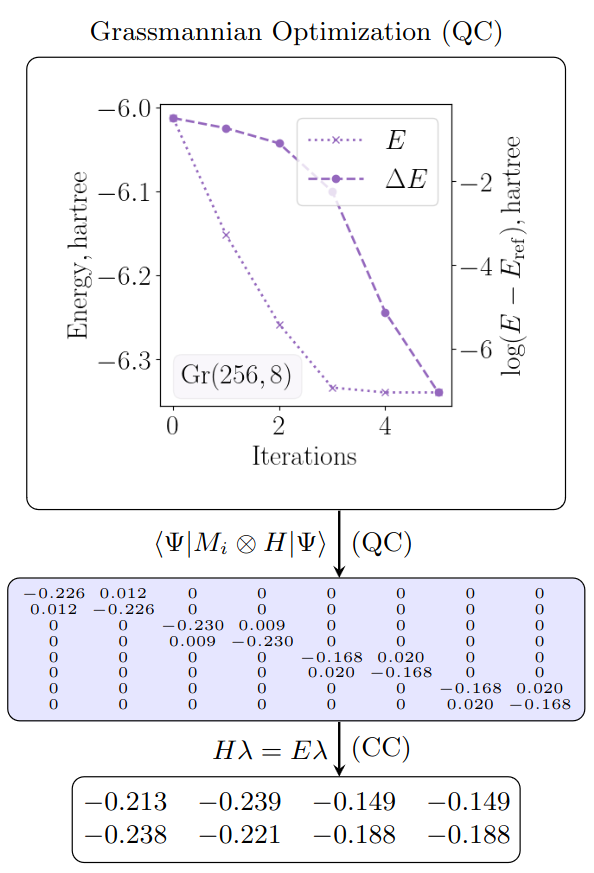}
\caption{Process of obtaining eigenvalues of a correlated molecule using quantum and classical techniques. First, we carry out an optimization of the Grassmannian manifold on a quantum computer for the subspace energy problem. The left shows the energy optimization and the right shows the log error with the lowest $8$ eigenvalues. We then measure the energies in the ancillary basis, resulting in an $8 \times 8$ matrix. Finally, we diagonalize the subspace to produce the final eigenvectors of the system. }
\label{fig:grass_flow}
\end{figure}

{\bf 3. Two-Step Grassmannian-Orthogonal Manifold Optimization: }The third solution translates the above problem to the quantum device. Namely, we replace the eigenvalue problem with a manifold optimization over $O(p)$. When $n$ is large, we might expect that the previous approach and this one have similar complexity. Whereas we might expect structure in the $n$-dimensional subspace (in cases where $H$ is a local operator) such that the cost of generating the $n \times n$ left tangent operator is sub-exponential, structure on the $p-$dimensional space might come from symmetries of the initial guesses or nearness of the solutions, leading to a block diagonal structure.

{\bf 4. Iterative Block-Diagonalization: } Finally, depending on the retraction, we can obtain the Stiefel optimization problem through a series of $\log_2 p$ Grassmannian optimizations. The initial optimization splits $H$ into a $p \times p $ block and a $n-p \times n-p$ block. For some retractions, we instead perform a Grassmannian-like optimization using a particular form of the Stiefel problem. Solving $\Gr(n,\frac{p}{2})$, with $\frac{p}{2}$ vectors taken from the $p$ block splits the $p-$block into two $\frac{p}{2} \times \frac{p}{2}$ blocks. Repeating this, we choose $\frac{p}{4}$ vectors from each $\frac{p}{2}$ block, yielding a different $\Gr(n,\frac{p}{2})$ problem. By performing $\log_2(p)$ optimizations, we can fully diagonalize the $p-$block, essentially performing matrix diagonalization via iterative block diagonalization. We detail this approach in the Appendix. 

\section{Discussion}

\textbf{Advancing Prior Variational Approaches}. The current approach is inspired by the unitary optimization involved in the subspace-search VQE approach (SSVQE)\cite{nakanishiSubspacesearchVariationalQuantum2019a}, for calculating excited states, as well as the contracted quantum eigensolver\cite{smartQuantumSolverContracted2021} for single eigenstates. In SSVQE, a single unitary is updated iteratively to locate each excited state. Additionally, in state-averaged approaches, multiple cost functions can be combined with variance constraints to prepare states\cite{yalouzStateaveragedOrbitaloptimizedHybrid2021}. In the contracted quantum eigensolver, we construct an iterative wavefunction based on information about the present state relevant in the stationary condition (i.e. the contracted equation). 

Following work on SSVQE, it was later demonstrated that multiple target states could be prepared simultaneously in a maximally entangled state \cite{xieOrthogonalStateReduction2022}, and more recently that the optimization could be performed at once, realized as the parallel or concurrent eigensolvers\cite{xuConcurrentQuantumEigensolver2023, hongQuantumParallelizedVariational2023}. During the preparation of this work, a similar ansatz was also extended to the contracted quantum eigensolver framework\cite{benavides-riverosQuantumSimulationExcited2023b}. Each of these methods focuses on either a static optimization of the unitary acting on the wavefunction or a weighted ensemble of states. While these later methods do offer the same logarithmic advantage of solving for $p$-eigenstates, similar problems to VQE are present, and these methods also do not fully exploit the nature of the Stiefel manifold tangent space when looking at Problem~\eqref{eq:prob_st}. It also appears in certain of these approaches that there is difficulty in optimizing the landscape of several states, which could reasonably be due to the use of gradient-based approaches and issues of parameterization. 

To compare the manifold approach with other variational approaches, we focus on the $p=1$ case, which is more well-documented and corresponds to the minimization of the unit sphere. We also refer to the notion of a trivialization, describing mappings on manifold problems that remove manifold constraints\cite{lezcano-casadoTrivializationsGradientBasedOptimization2019}. As noted by the authors of Ref.~\cite{lezcano-casadoTrivializationsGradientBasedOptimization2019}, Euclidean optimization over manifold parameterizations is ubiquitous with optimization problems and has a much longer history\cite{boumalIntroductionOptimizationSmooth2023}. 

The VQE can be seen as a static trivialization of the manifold problem. That is, the map $\phi$ of the manifold is a trivialization map which maps $\mathbb{R}^n$ to the manifold, and thus the cost function maps $\mathbb{R}^n$ to $\mathbb{R}$. The mapping itself induces a metric, which has been recognized in methods such as the quantum natural gradient\cite{stokesQuantumNaturalGradient2020,yamamotoNaturalGradientVariational2019a}. While we can consider the nature of the manifold, these methods can be more costly and are ultimately still limited by the choice of parameterization. While numerous variational quantum algorithms have been developed inspired by iterative ansatz frameworks, these methods still rely on Euclidean optimization over the parameter space.

The contracted quantum eigensolver, or CQE, is another variational approach that attempts to solve the contracted Schr{\"o}dinger equation (CSE) through the minimization of the residuals of the CSE in an iterative ansatz\cite{smartQuantumSolverContracted2021}. The standard form uses the anti-Hermitian form of the CSE, resulting in a two-body unitary operator. Later work showed one could perform a locally parameterized optimization (this is likened to a dynamic trivialization of the manifold), where at each step an approximate optimization step is performed\cite{smartAcceleratedConvergenceContracted2022a}. The optimizer does not consider the manifold structure of the state and has several heuristic elements, but in practice provides efficient performance, at the cost of potentially deep circuits. The CQE is derived relative to the CSE and has an intimate relation to reduced density matrix-based approaches\cite{
mazziottiAntiHermitianContractedSchrodinger2006}, and the exact exponential parametrizations of the wavefunction\cite{nakatsujiStructureExactWave2001}, and not necessarily the manifold optimization problem.

\textbf{Advantages of the Current Approach} 
The current framework is completely general in that we are focusing on incorporating the manifold structure of the quantum state into our optimization and providing tools related to Riemannian optimization. As a result, there is a higher degree of flexibility and generality for exploring various optimization strategies, establishing rates of convergences or adjusting tomographic schemes. The retraction, for instance, plays a key role in practical calculations, and while the exponential retraction should be more efficient on quantum devices, in practice other tools may be favorable.

The explicit tools provided in this work provide a mapping of tools in Riemannian optimization to the many-body eigenvalue problem, relevant for applications in quantum chemistry and physics and other state preparation-related problems. For multiple eigenstates, one benefit, noted in other work with a similar anastz, is that this wavefunction form allows for simultaneous estimation of a large number of excited states, despite using only the logarithm of that number of quantum states. The current work also fully incorporates the manifold structure of the problem into the optimization, inheriting rigorous convergence through Riemannian optimization techniques, as long as certain criteria in the implementations are met.

\textbf{Pathways Towards Efficient Realizations}.
Several challenges exist in the idealized form of the current approach. For large $p$, explicit evaluations of the cost function and gradient are offset by an exponentially small prefactor. That is, the measurement cost associated with $p=2^{10}$ results in a thousand-fold increase in measurement cost for nearly all quantities. However, this will depend heavily on the target application and the required tolerances. 

The maximally entangled state preparation step can also lead to deep circuits. As a simple example, while the maximally entangled state $|\Psi\rangle = \sum_i |i\rangle |i\rangle$ can be generated with linear scaling entangling gates, by perversely permuting the indices, very challenging mappings can be constructed, which may or may not have efficient implementations.

Finally, the tangent vector parameterization here spans the complete operator space, scaling exponentially in both state tomography and circuit depth costs. While the current numerical results are for few qubits and essentially allow for a complete representation of the tangent space, these are \emph{not scalable}. Implementing such a tangent vector is also challenging, and not feasible for near-term devices. Using an approach like shadow tomography\cite{huangPredictingManyProperties2020d}, the tomography of the gradient of the Grassmannian problem, involving the collection of operators $\hat{O}_i = [\hat{M}_{ij},\hat{H}]$, can be estimated efficiently to additive precision (depending on the shadow norm of the procedure), and it is possible an efficient unitary representation for the retraction could be developed as well. Again, the primary purpose in this work is to introduce a simple framework for exploring these problems through the lens of manifold optimization, and future work will explore these problems, particularly through the lens of reduced parameterizations. 

\section{Conclusion}
The eigenstate problem is critical for understanding many-body quantum systems and has been the focus of much research for multiple levels of quantum computing capabilities. Traditional variational quantum algorithms are hopeful candidates for finding efficient paths forward with near-term devices, but are limited and fundamentally focus on optimizing a parametrization of the state, rather than the state itself. As a result, problems related to barren plateaus, local minima, and finding optimal parameterizations are quite common.

In this work, we take a different approach to the eigenstate problem, realizing it as a manifold optimization problem over the Grassmannian and Stiefel manifolds. The Stiefel manifold represents a collection of $p-$orthogonal vectors, and the Grassmannian manifold represents the points on the Stiefel manifold which are equivalent under rotations of the $p-$vectors. For the Stiefel and Grassmannian manifolds, the tangent space can be parameterized in terms of operators acting on the state, and these operators can serve as proxies for elements of the tangent space at a given point. Using a quantum computer, we efficiently realize these manifolds and demonstrate retractions, vector transport, inner products, and measurement using quantum computational tools. Finally, this work demonstrates efficient and highly accurate many-body calculations and provides examples for solutions of the $p$-lowest eigenvalue problem using the Grassmannian manifold.

The current work utilizes the manifold structure of quantum states to find solutions to the many-eigenstate problem. That quantum states are manifolds allows us to efficiently represent certain embedded submanifolds and elements of the manifold on a quantum computer. The framework of manifold optimization here grants direct access to the state, offers rigorous convergence guarantees through Riemannian optimizers and provides sufficient and necessary criteria for eigenstate problems. Collectively, these represent an exciting new avenue for expanding and enabling eigenstate preparation techniques within quantum computation.

\section{Acknowledgments}
This work is supported by the U.S. Department of Energy Basic Energy Sciences (BES) under grant number DE-SC0019215, and the National Science Foundation RAISE-QAC-QSA under grant number DMR-2037783.

\bibliography{main.bib}

\appendix

\onecolumngrid

\section{Further Notes on Manifold Optimization}
In this section, we provide a few more details related to manifold optimization and the representation of tangent vectors using their actions. Additionally, we show derivations of the Riemannian gradient and Hessian expressions for the Stiefel and Grassmannian manifolds. 

\subsection{Alternate Representations of Grassmannian Tangent Vectors}
Let $Z$ be a tangent vector such that $Z=AX$. Then, we have $B = ZX^T = AXX^T$, which naturally leads to the assignment $L_X(Z) = B - B^T$. We then ask when does it hold that
\begin{equation}
(A - A^T )X = Z
\end{equation}
and we can replace the left action $L_X(AX)$ with $A - A^T$? We begin by decomposing a matrix $O$ into the projection on $X$ and its complement, via the subspace resolution of the identity $I_n = \X + \X_\perp$. Associating the vectors of $X$ with the standard basis, we can represent $O$ as a block-diagonal matrix (with block dimensions $p$ and $n-p$):
\begin{align}
A &= \begin{pmatrix}
\X O \X & \X O \X_\perp \\
\X_\perp O \X & \X_\perp O \X_\perp
\end{pmatrix}.  
\end{align}
The block diagonal form of $B$ can be expressed as:
\begin{equation}
B = \begin{pmatrix}
0 & 0 \\
\X_\perp B \X & 0
\end{pmatrix}.
\end{equation}
which yields a left action $L_X(V)$:
\begin{equation}
L_X(V) = B - B^T = \begin{pmatrix}
0 & -\X B^T \X_\perp \\
\X_\perp B \X & 0 
\end{pmatrix}.
 \end{equation}

Regarding $A$, we then have the following result:

\begin{theorem} Given a factorization $Z=AX$ when $Z \in T_X \Gr$, the matrix $A - A^T$ serves as a suitable left-action if and only if $\X A = 0$. Equivalence of inner products holds when $\X_\perp (A-A^T) \X_\perp =0$. 
\end{theorem} 
\begin{proof}
If $\X A =0$, then $A = \X_\perp A $ and $A^T X = A^T \X_\perp X = 0$. This implies that $(A-A^T)X = Z$. Conversely, if $ A - A^T = Z$, then we can write the block-diagonal form:
\begin{align}
(A - A^T)X &= \begin{pmatrix}
0_{p,p} & \X  (A-A^T) \X_\perp \\
\X_\perp  (A-A^T) \X & \X_\perp (A-A^T) \X_\perp 
\end{pmatrix}  \begin{pmatrix}
I_p  \\ 0_{n-p,p} 
\end{pmatrix} \\ &= \begin{pmatrix} 0_{p,p}  \\ \X_\perp (A-A^T) X
\end{pmatrix}  = \begin{pmatrix} 0_{p,p} \\
C - \X_\perp A^T X 
\end{pmatrix}
\end{align}
which implies that $\X_\perp A^T \X = 0$, which combined with $\X_\perp A^T \X=0$ fulfills $A^T  \X = 0 $. The requirement for the inner product is based on those terms needing to be 0 in order to recover the proper inner product. 
\end{proof}
The above theorem implies the following structure of $A$:
\begin{align}
A &= \begin{pmatrix}
0 & 0  \\
\X_\perp A \X & \X_\perp A \X_\perp
\end{pmatrix}.  
\end{align}
We can then use this to construct an affiliated function of a \emph{generic} $n \times n$ operator $A$ to use as a representation on $X$:
\begin{equation}
J_X(A) = \X_\perp A - A^T \X_\perp  
\end{equation}
or equivalently, 
\begin{equation}
J_X(A) = \begin{pmatrix}
0_{p,p} & -\X A^T\X_\perp \\
\X_\perp A \X & X_\perp (A-A^T) X_\perp
\end{pmatrix}
\end{equation}
The exact form of $J_X(A)$ can be found as follows:
\begin{align}
J_X(A) &= \X_\perp A - A^T \X_\perp   \\
&= [\X_\perp, \sy(A) ] + \{ \X_\perp, \sk(A) \} \\
&= [\sy(A) , \X] + 2~\sk(A) - \{ X, \sk(A) \}  
\end{align}
If we compare with the left action of $AX$, we have:
\begin{align}
L_X(P_X^\Gr(AX)) &= [\sy(A),\X]  + \sk(A)\X + \X \sk(A) - 2 X \sk (A) \X \\
&= [\sy(A),\X] + \X_\perp \sk(A) \X + \X \sk(A) \X_\perp   
\end{align}
The symmetric contribution is equivalent, and so it is simply a question of the skew-symmetric contribution. The $L_X(P_X^\St(AX))$ is designed in such as way that $X_\perp^T L X_\perp = 0$, whereas this is not the case for $J_X$, and this difference can be expressed as $\X_\perp A \X_\perp$. 

\subsection{Actions on the Stiefel Manifold}

Similar to the Grassmannian case, for a vector $Z \in T_X \St(n,p)$, we have according to the projection with respect to $X$:
\begin{equation}
B = Z X^T = \begin{pmatrix}
\X B \X  & 0_{p,n-p} \\
\X_\perp B \X & 0_{n-p,n-p}  
\end{pmatrix}
\end{equation}
Then, $B - B^T$ yields:
\begin{equation}
L_X(Z)  = \begin{pmatrix}
X (B - B^T) X^T  & -\X B^T \X_\perp \\
\X_\perp B \X & 0_{n-p,n-p}  
\end{pmatrix}
\end{equation}
which when matched with $R_X(Z)$ yields the tangent vector:
\begin{equation}
Z = L_X(Z) X - X R_X(Z).
\end{equation}

If we consider a form $Z = AX = XR + X_\perp C $, then we run into a similar problem as before. Namely, 
\begin{align}
A &= \begin{pmatrix}
X R X^T  & \X A \X_\perp    \\
\X_\perp C X^T & \X_\perp A \X_\perp
\end{pmatrix}
\end{align}
namely that $\X A \X_\perp$ is non-zero. We can again use the Grassmannian projection with $J$, and this more or less solves our problem. 
\begin{equation}
\X_\perp A = \begin{pmatrix}
0  & 0  \\
\X_\perp C X^T & \X_\perp A \X_\perp
\end{pmatrix}. 
\end{equation}
The resulting factorization is then:
\begin{equation}
Z = J_X(Z) X + X R_X(Z). 
\end{equation}

\subsection{Riemannian Metric for Tangent Space Vectors and Actions}
Here we expand on the proof in Lemma \eqref{lem:metric}. The Euclidean inner product of two tangent vectors of the Stiefel manifold is:
\begin{align}
\langle Z,  W \rangle &= \Tr C_Z^T C_W + \Tr A^T_Z A_W .
\end{align}
In earlier work the apparent double counting of the skew-symmetric degrees of freedom led to the canonical Stiefel manifold\cite{edelmanGeometryAlgorithmsOrthogonality1998a}, but for our purposes we use the metric associated with the Euclidean trace product. 
The matrices $L_X(V)$ and $J_X(Z^T X)$ have the following decompositions: 
\begin{align}
L_X(Z) &= X_\perp C X^T - X C^T X_\perp + 2 XAX^T ,\\
J_X(Z) &= X_\perp CX^T - X C^T X_\perp  .
\end{align}
For $Z,W \in \St(n,p)$, we have:
\begin{align}
\langle Z, W \rangle &= \Tr (X^T J_Z^T + A_Z^T X^T)(J_W X + X A_W) \\
&= \Tr J_Z^T J_W \X + A_X^T A_W  \\
&= \Tr J_Z^T J_W - X^T_\perp J_Z^T J_W X_\perp + A_X^T A_W \\
&= \Tr J_Z^T J_W - C_Z C_W^T + A_X^T A_W
\end{align}
implying that $\Tr J_Z^T J_W = 2 \Tr C_Z C_W^T$. For $L_X(V)$, we have:
\begin{align}
\begin{split} \langle L_X(V), L_X(W) \rangle &= \Tr (- X_\perp C_V X^T + X C^T_V X_\perp + 2 XA_Z^T X^T )\\
&~~~~~\times (X_\perp C_W X^T - X C_W^T X_\perp + 2 XA_WX^T )\end{split} \\
&= \Tr (C_V C_W^T + C_V^T C_W + 4 A_V^T A_W) \\
&= 2 \langle Z, W \rangle_X + 2\langle A_V, A_W\rangle_X 
\end{align}
Which yields:
\begin{equation}
\langle Z, W \rangle = \frac{1}{2} \langle L_X(V), L_X(W) \rangle - \langle A_X(V), A_X(W)\rangle 
\end{equation}

\subsection{Tangent Operator Identities}
In light of the previous sections, we derive some expressions which are useful in manipulating more complicated expressions.

Given that $L = B - B^T$, it follows that:
\begin{align}
\X L &= L \X_\perp \\
L \X &= \X_\perp L \\
L &= L\X + \X L  
\end{align}
Finally, we note two identities that the left action possess. 
\begin{align}
L &= XX^T L + L XX^T - XX^T L XX^T \\
J &= XX^T J + J XX^T \\
 A_X(Z) &= \frac{1}{2} X^T L_X(Z) X^T \\
 L_X(Z) &= J_X(ZX^T) + 2 X A_X(Z) X^T \\
0 &=  X^T J_X(Z X^T) X.
\end{align}
We utilize some of these in later expression, and note that they can be used for other derivations involving the Stiefel and Grassmannian manifolds as well. 

\subsection{Riemannian Hessian for the Subspace Problem}
Given that the Hessian of Eq.~\eqref{eq:prob_gr} is 
\begin{align}
{\rm Hess}f(X)[Z] &=  (H L_Z   - \X H L_Z - L_Z \X H )X  
\end{align}
we consider the perpendicular contribution of the left factorization $A_{\rm Hess}$, namely:
\begin{align}
\X_\perp A_{\rm Hess} \X_\perp &= \X_\perp (HL_Z - \X H L_Z - L_Z \X H) \X_\perp \\ &=  \X_\perp (HL_Z - L_ZH + \X L_Z H ) \X_\perp \\
&= \X_\perp [H,L_Z] \X_\perp 
\end{align}
which is a symmetric quantity and thus vanishes in $J_X$. Then, our expression of interest is:
\begin{align}
L_X(A_{\rm Hess}X) &= J_X(A_{\rm Hess}) \\
&= HL_Z + L_ZH - L_ZXH - HXL_Z - XHL_Z - L_ZHX \\
&= HL_ZX + XL_ZH - XHL_Z - L_ZHX \\
&= [[H,L_Z],X]. 
\end{align} 
Additionally, the same form can be derived using $L_X({\rm Hess})$ instead of $J_X$. 

\subsection{Riemannian Gradient and Derivative for the Stiefel Problem}
The gradient and Hessian of Stiefel manifold presents more involved derivations, and are detailed below. 
\subsubsection{Riemannian Gradient}

Given the eigenvalue equation Prob.~\eqref{eq:prob_st}, smooth extension of the gradient is:
\begin{equation}
{\rm grad} \bar{f}(X) = H X K.
\end{equation} 
The projection onto the Stiefel manifold is:
\begin{align}
{\rm grad} f(X) = HXK + X^T \sk (X^T HXK)
\end{align}
We can find $L_X$ and $J_X$ as:
\begin{align}
L_X({\rm grad} f(X)) &= L_X(HXK) \\
&= HXKX^T - X KX^T H \\
&= [H, \X_K]
\end{align}
\begin{align}
J_X(HXKX^T) &= X_\perp  HXKX^T - X KX^T H X_\perp \\
&=  [H, \X_K] - XX^T H X KX^T + XKX^T H XX^T \\
&= [H, \X_K] - X [X^T H X, K] X^T 
\end{align}
Alternatively, we also have:
\begin{align}
J_X(HXKX^T)&= H \X_K - \X_K H - \X H \X_K + \X_K H \X \\
&= [H,\X] \X_K  + \X_K [H,\X] \\
&= \{[H,\X],\X_K\} \\
&= \{G, X_K \} 
\end{align}
where $G$ is the Grassmannian gradient expression. The right action is given as:
\begin{equation}
A_X(HXK) = \sk(X^T HXK) = \frac{1}{2}[X^T H X, K]
\end{equation}
Here, $L_X(V)$ likely is easier to realize than $J_X(VX^T)$, which is reasonable given that $J$ requires the Grassmannian projection.
\subsubsection{Riemannian Hessian}\label{sec:stiefhess}
The Hessian of Prob.~\eqref{eq:prob_st} appears to have a more involved form, partially due to the parametrization of $V$ as well as the presence of the right actions. So far, we have not identified a computationally simple Hessian form, as in the Grassmannian case. For a tangent vector $V \in T_X \St(n,p)$, we can express the Hessian as:
\begin{equation}
{\rm Hess}~g(X)[V] = P_X^\St(HVK - V \sy X^T  H XK ) 
\end{equation}
Let $V =  L_X(V) X - XA_X(V)$. Denoting the left and right actions as $L_V$ and $A_V$ respectively, and letting $E$ be the subspace energy $E = X^T H X $, we have:
\begin{align}
A_X({\rm Hess}~g(X)[L_X(V) X]) &= \sk(X^T HVK - X^T V \sy (X^T  H XK )) \\
&= \sk(X^T H L_V X K  - X^T L_V X \sy (E K)) \\
&= \frac{1}{2} (X^T HL_V XK +  K X^T L_VH X) - \sk(A_V E K + A_V K E) 
\end{align}
\begin{align}
A_X({\rm Hess}~g(X)[XA_X(V)]) &= \sk(X^T HVK - X^T V \sy( X^T  H XK )) \\
&= \sk(X^T H XA_V K ) - \sk(A_V \sy (EK )) \\
&= \frac{1}{2} (E A_V K + K A_V E) - \frac{1}{4}\{ \{A_V, E \}, K\}  .
\end{align}

The most challenging part here is finding the contribution including $HL$, stemming from the skew symmetric portion of $V_L$. We can isolate this contribution in the tangent vector:
\begin{align}
X \sk(X^T HL_VXK)  &= \frac{1}{2}( XX^T HL X KX^T X + X KX^T LH X ) \\
&= \frac{1}{2} (\X HL \X_K + \X_K LH \X ) X,
\end{align}
though we do not find either form to be truly satisfactory. Looking at the left action of the Hessian, we find:
\begin{align}
L_X({\rm Hess}~g(X)[L_X(V)X]) &=   HL_V \X_K - L_V X \sy (X^T  H XK)X^T +  \X_K L_V H -  X \sy(X^T  H XK) X^T L_V \\ 
\begin{split} &=  [[H, L_V],\X_K] + L_V H \X_K + \X_K H L_V  - \frac{1}{2} \{L_V,(\X H \X_K+\X_K H \X) \}
\end{split}
\end{align}
The symmetric portion of $\X H \X_K$ has the following form:
\begin{align*}
\X H \X_K + \X_K H \X &= [\X,H] \X_K + \X_K [H,\X] + H \X_K + \X_K H \\
&= [\X_K, G] + \{ H, \X_K \} 
\end{align*}
where $G$ is the Grassmannian gradient, $G= [H, \X]$. Using this, we have:
\begin{align}
\begin{split} L_X({\rm Hess}~g(X)[L_X(V)X]) 
 &=  [[H, L_V],\X_K] + L_V H \X_K + \X_K H L_V  - \frac{1}{2} \{L_V,([\X_K, G] + \{ H, \X_K \} \} 
\end{split} \\
\begin{split} &=  [[H, L_V],\X_K] + L_V (H \X_K - \frac{1}{2}\{H,\X_K \} - \frac{1}{2}[\X_K, G]) \\ &~~~ + (\X_K H-\frac{1}{2}\{ H, \X_K \} - \frac{1}{2}[\X_K, G] )  L_V  
\end{split}
\\
&= [[H, L_V],\X_K] + \frac{1}{2}L_V( [H,\X_K]-[\X_K,G]) - \frac{1}{2} ([H,\X_K] + [\X_K, G]) L_V \\
&= [[H, L_V],\X_K] + \frac{1}{2} [L_V, G_K] - \frac{1}{2}\{L_V, [\X_K, G] \} 
\end{align} 
Finally, the right contribution on the Hessian yields:
\begin{align*}
L_X({\rm Hess}~g(X)[XA_X(V)]) &= H X A_V K X^T - X A_V (\sy X^T  H XK)X^T + X K A_V X^T H - X \sy(X^T H X K) A_V X^T \\
&= H X A_V K X^T + X K A_V X^T H - \frac{1}{2} X \{ A_V, \{ E, K \} \}  X^T   
\end{align*} 
The last portion of this can be factored as a right action:
\begin{align*}
 \frac{1}{2} X \{ A_V, \{ E, K \} \}  X^T   X  &= \frac{1}{2} X \{ A_V, \{ E, K \} \}.
\end{align*}
The total form of the Hessian for $f_S(X) = {\rm Tr} X^T H X K$ is then given by:
\begin{equation}
\begin{split}{\rm Hess} f_S(X)[V] &= \Big( [[H, L_V],\X_K] + \frac{1}{2} [L_V, G_K] - \frac{1}{2}\{L_V, [\X_K, G] \}+ H X A_V K X^T + X K A_V X^T H \Big)X \\
&~~ -X \Big(\sk( K X^T HL X) - \frac{3}{4} \{ A_V, \{ E, K \} \} -\frac{1}{2}(E A_V K + K A_V E ) \Big) 
\end{split}
\end{equation}
We suspect that a more compact form of the Hessian exists for $f_S$, though our work indicates the Grassmannian manifold is likely more suitable for near-term oriented applications. Additionally, further problems regarding the complexity of the operator forms need to be addressed for the Stiefel manifold.
\section{Iterative Grassmannian and Stiefel Optimization}
Here, we provide a further explanation of the iterative Grassmannian optimization approach which allows us to prepare up to $p-$distinct eigenvectors. Let $p$ denote the target number of states. For the zeroth iteration, we perform an optimization over $\Gr(n,p)$ with initial state $X$. There are two cases we consider. In the first, we assume that it is easy to implement the retraction different to multiple sets of vectors. The second considers when this is not possible. 

\textbf{Iterative Grassmannian Optimization.}
If our retraction can be applied to vectors of $X_\perp$ as well as $X$, then we consider a retraction on $H$ as a retraction acting on $H$. That is, given $Y = R_X(V)$ (representing potentially several retractions), we can write:
\begin{align*}
Y^T H Y &= X^T R_X(V)^{-1} H R_X(V) X \\
&= X^T H^0 X 
\end{align*}
where $H^0$ is block diagonal:
\begin{equation}
H^0 = \begin{pmatrix}
\X_\perp H^0 \X_\perp & 0 \\
0 & \X H^0 \X 
\end{pmatrix}. 
\end{equation}
Note, for the retractions generated by skew-symmetric matrices considered in this work, which are applied to the entire $n$-dimensional system, this can easily be done.

Now, we focus on the diagonalization of the $p \times p$ block, $\X H^0 \X$. Assuming $p$ is a power of 2, we need a sequence $\mathcal{S}$ of $\log_2 p$ partitions, where for each partition, exactly half of the vectors can be found in every previous sequence. The set of length $\log_2 p$ bits where the $k-$th element is 1, denoted as $P_k$, give such a sequence. For example, given $p=8$, we have $P_1 = \{100, 101, 110, 111\}$, $P_2 = \{010, 011, 110, 111 \}$, and $P_3 = \{001, 011, 101, 111 \}$. $P_2$ shares $110$ and $111$ with $P_1$, and $P_3$ shares $101$ and $111$, and $011$ and $111$ with $P_1$ and $P_2$ respectively. The sequence then is $\mathcal{S} = (P_1, P_2, P_3)$. 

We choose $X^1$ to consist of vectors with indices from our set $P_1$, and then we perform the Grassmannian optimization on $\Gr(n,p/2)$. The resulting retraction $Y^1 = R_{X^1}(V^1)$, yields:
\begin{align*}
(Y^1)^T H^0 Y^1 &= (X^1)^T (R_{X^1}(V)^{-1} H^0 R_{X^1}(V)) X^1 \\
&= X^T H^1 X .
\end{align*}
Because $H^0$ was block diagonal, the only coupling of vectors was within the $p \times p$ block, and so we have effectively block diagonalized our matrix again:
\begin{equation}
H^1 = \begin{pmatrix}
\X_\perp H^0 \X_\perp & 0 & 0 \\
0 & \X^1_\perp H^1 \X^1_\perp & 0 \\
0 & 0 & \X^1 H^1 \X^1 
\end{pmatrix}. 
\end{equation}
where $\X^{(1)} = X^{(1)} (X^{(1)})^T$ and $\X_\perp^{(1)} = I_p - \X^{(1)}$. In the next iteration, we take $X^{(2)}$ to be vectors with indices given by $P_2$, and by our choice of sequences, each block will be split into two blocks until only a single element results. The result is a complete diagonalization of $H$. 

\textbf{Grassmannian Optimization through Stiefel Optimization.}
If the retraction \emph{cannot} easily be applied to vectors of $X_\perp$ as well as $X$, then we still begin with an optiimzation over $\Gr(n,p)$ to produce $Y = R_X(V)$, and we still have effectively block-diagonalized our matrix. 

Next, we perform a Stiefel optimization on the same set of vectors, with a particular $K$ matrix. In particular, let $K$ be defined as: 
\begin{equation}
(K_0)_{ii} = \begin{cases} +1~{\rm if}~i \in P_0 \\ -1~{\rm else}
\end{cases}
\end{equation}
which corresponds to the unitary Pauli matrix $K_j = -Z_j$. Minimization over $\St(n,p)$ here corresponds to the minimization of $p/2$ vectors and the maximization of $p/2$ vectors \emph{within} the $p$-subspace, resulting in the same block diagonalization as in the Grassmannian case.

This yields a simpler form of the Grassmannian minimization, involving only the left action. Letting $X_S$ denote elements of the Stiefel manifold, we have:
\begin{equation}
X_S^T L_{X_S}(V) X_S = X_S^T  V - V^T X_S^T = L_{X_S}
\end{equation}
which according to Eq.~\eqref{eq:left_right}, implies that we only need the left action or right action of the state to implement the tangent vector. Thus, it is possible that given some sparse $L_{X_S}$, we can avoid evaluating the subspace energy matrix.  

The primary advantage here over a traditional Stiefel optimization is that the norm of $K$ is well defined (i.e. we do not need a range of exponential scaling values of $p$), and it also result in simpler Riemannian gradient and Hessian forms based on the block diagonal structure of $E$ (which become more sparse). Regardless, further scaling problems do exist, and will be the subject of further work. 

\textbf{Challenges in Implementation}. There are several limitations in this approach, notably that the condition of block diagonalization assumes that it is effectively carried out to a high precision, which is not amenable for current quantum devices. In the case that a large threshold is used, projections would have to be implemented to maintain the pseudo-block diagonal forms, which could potentially accumulate errors. Also, it is not necessarily the case that this method scales more favorably than the above methods, depending on the complexity of the Grassmannian optimization.

\end{document}